\documentclass[a4paper,11pt]{article}
 
\usepackage{microtype}
\usepackage{amssymb,amsmath,amsthm} \usepackage{graphicx}
\usepackage{xspace} \usepackage{xypic} \usepackage[numbers]{natbib}
\usepackage{hyperref,color}
\usepackage[T1]{fontenc}
\usepackage[mathlines]{lineno} 
\usepackage{fullpage}
\usepackage{pdfpages}

\title{Assigning Weights to Minimize the Covering Radius in the
  Plane\footnote{This research was supported by the MSIT(Ministry of
    Science and ICT), Korea, under the SW Starlab support
    program(IITP--2017--0--00905) supervised by the IITP(Institute for
    Information \& communications Technology Promotion) and the NRF
    grant 2011-0030044 (SRC-GAIA) funded by the government of Korea.}}

\author{Eunjin Oh\thanks{Max Planck Institute for Informatics, Saarbr\"ucken, Germany, Email: \texttt{eoh@mpi-inf.mpg.de}} \and Hee-Kap Ahn\thanks{Department of Computer Science and
    Engineering, Pohang, POSTECH, Korea, Email: \texttt{heekap@postech.ac.kr}}}

\newtheorem{theorem}{Theorem} 
\newtheorem{lemma}[theorem]{Lemma}

\newtheorem{corollary}[theorem]{Corollary}

\newcommand{\eps}{\ensuremath{\varepsilon}}
\newcommand{\wnw}{\ensuremath{W_{1}^{w}}}
\newcommand{\wone}{\ensuremath{W_1}}

\begin{document}
\date{}
\maketitle
\begin{abstract}
  Given a set $P$ of $n$ points in the plane and a multiset $W$ of $k$
  weights with $k\leq n$, we assign each weight in $W$ to a distinct point in $P$ to minimize
  the maximum weighted distance from the weighted center of $P$ to any
  point in $P$.  In this paper, we give two algorithms which take
  $O(k^2n^2\log^3 n)$ time and $O(k^5n\log^3k+kn\log^3 n)$ time,
  respectively.
  For a constant $k$, the second algorithm takes only $O(n\log^3n)$ time, which is near
  linear.
\end{abstract}

\section{Introduction}
Consider a set of robots lying at different locations in the plane.
Each robot is equipped with a locomotion module so that it moves to
a nearby facility to recharge its battery and returns to its original
location.  We want to place a recharging facility for the robots such
that the maximum travel time of them to reach the recharging facility
is minimized.  The Euclidean center of the robot locations may
not be a good location for the recharging facility if the robots have
huge differences in their speeds. For instance, consider three robots,
each lying on a different corner of an equilateral triangle. If one of
them has a much smaller speed compared to the speeds of the other two
robots, the best location is very close to the corner where the
low-speed robot lies.  In this case, the recharging facility must be located
at a \textit{weighted center} of the robots by considering their
speeds as weights of their placements.

In the weighted center problem, each input point is
associated with a positive weight. The \textit{weighted distance}
between an input point and a point of the plane is defined to be their
distance 
divided by the associated weight of the input point. Then the point of the plane that 
minimizes the maximum weighted distance to input points is
the center of the weighted
input points, which we call the \textit{weighted center}.
Dyer~\cite{weighted_center} 
gave a linear-time algorithm to compute the weighted center of a set of weighted points in the plane.
Clearly, the weighted center coincides 
with the (unweighted) center if the associate weight is the same for every
input point. 

Imagine now that we are allowed to \textit{reassign} the
locomotion modules of the robots.
Or, if the mobile robots are identical, except their speeds,
we are allowed to \textit{relocate} the robots.
A relocation of robots (or a reassignment of locomotion modules)
may change the weighted center and the maximum travel
time for mobile robots to reach the weighted center.
In other words, a clever assignment of robots (or their locomotion modules)
to given locations may decrease the objective function value.

In this paper, we formally define this relocation problem and present
algorithms for it. The \emph{weight assignment problem} is
defined as follows:
given an input consisting of a set $P$ of $n$ points in the plane
and a multiset 
$W=\{w_1,w_2,\ldots,w_k\}$ of $k$ weights of positive real values with $k\leq n$,
find an assignment of the weights in $W$ to input points such that the maximum 
weighted distance from the 
weighted center to input points is minimized. 
We assume that every input point of $P$ is assigned the default weight $1$.
We assign the $k$ weights to $k$ points of $P$ 
  such that every weight of $W$ is assigned to one of the $k$ points,
  each of the $k$ points gets one weight of $W$, and the remaining $n-k$
points of $P$ have the default weight
$1$. We call this an \textit{assignment of weights} of $W$
to $P$.

We regard an assignment of weights as a function.
To be specific, for an assignment $f$ of weights, let $f(p)$ denote
the weight assigned to a point $p \in P$.
We use $c(f)$ to denote the weighted center of $P$ by 
the assignment $f$, and  
call the maximum weighted distance from 
$c(f)$  to input points the \textit{covering radius} of the assignment $f$
and denote it by $r(f)$.

Obviously, there are $\binom{n}{k}$ different combinations of selecting
$k$ points from $P$ and $k!$ different ways of assigning the $k$ weights
to a combination, and therefore
there are $\Theta(n^kk!)$ different 
assignments of weights. 
Our goal is to find an assignment $f$ of the weights of $W$ to the points of $P$
that minimizes the covering radius $r(f)$ over all possible 
assignments of weights.

\subparagraph{Related Work.}  As mentioned earlier,
Dyer~\cite{weighted_center} studied the weighted center problem for a
set of weighted points in the plane. He reformulated the problem as an
optimization problem with linear inequalities and one quadratic
inequality. Then he gave a linear-time algorithm to compute the 
weighted center using the technique by Megiddo~\cite{Megiddo-linear}.
Later, Megiddo~\cite{Megiddo-balls} gave another linear-time algorithm
for the same problem using a different technique.

In contrast, to our best knowledge, no algorithm is known for the
weight assignment problem while there are works on several related
problems.  In the \emph{inverse 1-center problem} on graphs, we are
given a graph and a target vertex, and we are to increase or decrease
the lengths of edges of the graph so that the target vertex becomes a
center of the modified graph.  A center of a graph is a vertex
  of the graph which minimizes the maximum distance from other
  vertices in the graph. Notice that a modification in the lengths of the edges
  is equivalent to assigning additive weights to the edges.
The goal is to minimize the modification of the lengths.  
Cai et
al.~\cite{inverse_NPhard} showed that this problem is NP-hard on a
general directed graph.  Recently, Alizadeh and
Burkard~\cite{inverse_tree} gave an $O(n^2r)$-time algorithm for this
problem on a tree, where $r$ is the compressed depth of the tree.  A
variant of this problem is the \emph{reverse 1-center problem}, in
which we are to decrease the lengths of edges of the graph under a
given budget so that the maximum distance from a predetermined target vertex
to any vertex of the graph is minimized. This problem is also known to be NP-hard even on a
bipartite graph~\cite{reverse_NPhard}, and there is an
$O(n^2\log n)$-time algorithm on a tree by Zhang et
al.~\cite{reverse_tree}.

Our weight assignment problem is closely related to the weight
balancing problem which was studied by Barba et
al.~\cite{weight_balancing}.  The input consists of a simple polygon,
a target point inside the polygon, and a set of weights.  The goal is
to put the weights at points on the boundary of the polygon
so that the barycenter (center of mass) of the weights coincides with
the target point.  They showed the existence of such a placement of
weights under the condition that no input weight exceeds the sum of
the other input weights.  They also gave an algorithm to find such a
placement in $O(k + n\log n)$ time, where $k$ is the number of the
weights and $n$ is the number of the vertices of $P$.  Our problem can
be considered as a discrete version of this problem, but with a
different criteria (minimizing the covering radius), in the sense that
we place the weights on predetermined positions.

\subparagraph{Our Result.}  In this paper, we present two algorithms
 that compute an assignment $f$ of weights minimizing $r(f)$.  
 The first algorithm returns an optimal assignment of weights in $O(k^2n^2\log^3 n)$
time using $O(kn)$ space.  
  We first observe that the number of all possible weighted centers
  is $\Theta(k^3n^3)$ because a weighted center 
  is determined by at most three weighted points.
  To achieve the $O(k^2n^2\log^3 n)$-time algorithm, we use an
  algorithm deciding, for a given real value $r>0$, if there is an assignment
  $f$ of weights with $r(f)\leq r$. This algorithm makes use of the
  fact that there is a point $c$ satisfying $d(p,c)/f(p)\leq r$ for all
  points $p\in P$
  if and only if $r(f)\leq r$, where $d(p,c)$ denotes the Euclidean distance between $p$ and $c$.
  Moreover, there is a point $p\in P$ and a weight $w\in W\cup\{1\}$
  such that the circle $C$ centered at $p$ with radius $wr$ contains such a point $c$.
  Thus, the decision algorithm checks if there is such a point $c$ by
  considering each of the $O(kn)$ intervals on $C$ induced by the 
  concentric circles centered at each point in $P\setminus\{p\}$ with radius
  $w'r$ for each distinct weight $w'$. 

Then the overall algorithm finds an optimal assignment of weights
  and its covering radius by
  applying parametric search using the decision algorithm. In doing so,
  the algorithm computes the combinatorial structure of
  the arrangement of the circles $C_p$ for an optimal weight assignment $f^*$
  without knowing $f^*$ and its covering radius $r^*$,
  where $C_p$ is a circle centered at $p$ with radius $f^*(p)r^*$ for every
  $p\in P$.
  This is done by computing, for each pair of a point $p\in P$ and
  a weight $w\in W\cup\{1\}$, the sorted list of intersection points
  on the circle centered at a point $p$ with radius $wr^*$
  by the concentric circles 
  centered at each point in $P\setminus\{p\}$ with radius
  $w'r^*$ for $w'\in W\cup\{1\}$
  without knowing the optimal covering radius $r^*$. This takes $O(k^2n^2\log^3 n)$ time.
  Once the combinatorial structure is constructed,
  the algorithm computes an optimal assignment of weights and its covering
  radius in $O(k^2n^2\log^2n)$ time using $O(kn)$ space.
  This is done by sorting $O(k^2n^2)$ candidate radii defined by three circles
  of the arrangement and applying binary search on the sorted list using
  the decision algorithm.

The second algorithm computes an optimal assignment of weights in
$O(k^5n \log^3 k + kn\log^3 n)$ time assuming that 
every weight in $W$ is at most $1$.
The second algorithm is faster than the first algorithm when $k$ is sufficiently small
$(k=o(n^{1/3}))$.  Moreover, it takes only $O(n\log^3 n)$ time when
$k$ is a constant.

A merit of our algorithms is that they are based on useful 
geometric intuition and are easy to implement though they use  
parametric search, the optimization technique developed by 
Megiddo~\cite{parametric}.
The parametric search technique is an important tool for solving
many geometric optimization problems efficiently, but algorithms based on it
are often not easy to be implemented~\cite{Agarwal1995}. 
A main difficulty lies in computing the roots of the polynomials exactly whose signs 
determine the outcome of the comparisons made by the algorithm. 
However, as we will see later,
in our algorithms, such a root is a covering radius
of at most three weighted points, so we can compute it easily
 instead of resorting to complicated methods.

\section{Preliminaries}
\label{sec:example}
For any two points $p$ and $q$ in the plane, we use $d(p,q)$
  to denote the Euclidean distance between $p$ and $q$.
  For a weighted point $p$ in the plane, let $w(p)$ denote the weight of $p$.
  For a weighted point $p$ and an unweighted point $q$ in the plane, their weighted
  distance is defined as $d(p,q)/w(p)$. For a set of weighted points, the weighted
  center of the weighted points is a point in the plane that minimizes the maximum
  weighted distance to the weighted points.
\begin{figure}
  \begin{center}
    \includegraphics[width=.8\textwidth]{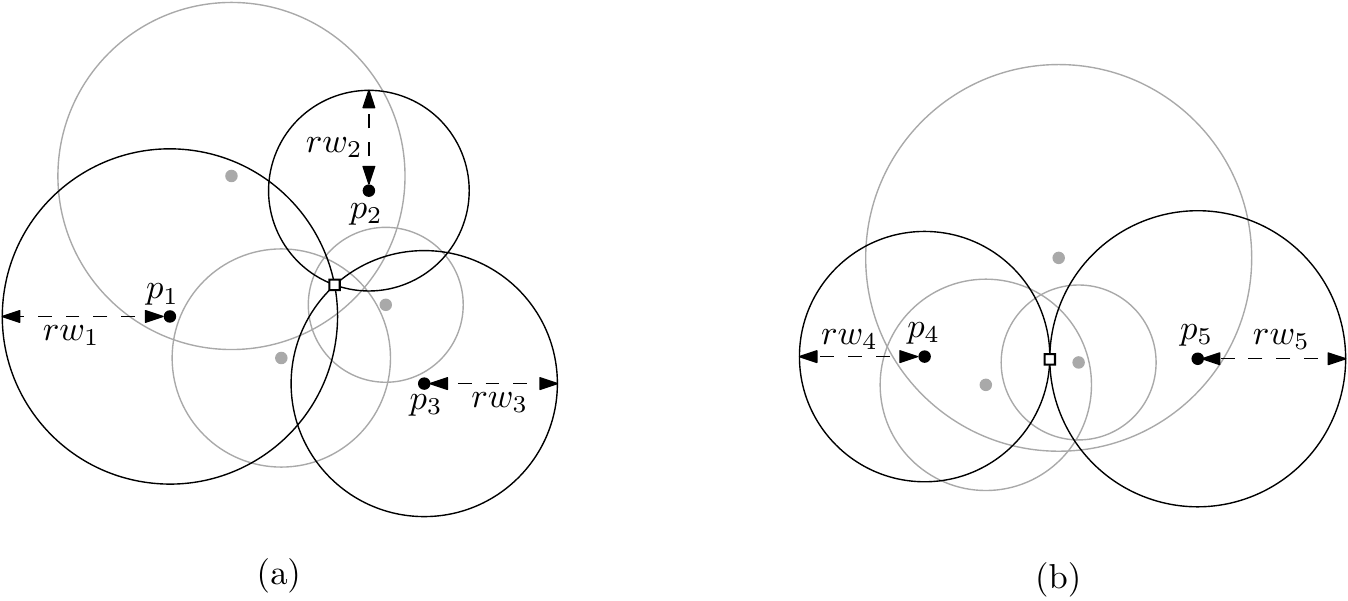}
    \caption {(a) The weighted center of three points $p_i$ with weight
        $w_i$ for $i=1,2,3$ is the weighted center of whole weighted points.
        (b) The weighted center of two points $p_4$ with weight $w_4$
        and $p_5$ with weight $w_5$ is the weighted center of whole weighted points.}
    \label{fig:arrangement_structure}
  \end{center}
\end{figure}

Let $P$ be a set of $n$ points in the plane, and let
$W=\{w_1,w_2,\ldots,w_k\}$ be a multiset of $k$ weights of positive real values.
To ease the description, we assume that $k<n$ unless stated otherwise.
In case that $k=n$, there is no difference in the proposed algorithm, except that
every point of $P$ is assigned a weight of $W$. 
One way to deal with our problem is to find the weighted
centers for all possible assignments of weights and choose 
the one with the minimum covering radius.  Although there are 
$\Theta(n^kk!)$ different assignments of weights, 
there are only $O(k^3n^3)$ different weighted
centers. This is because a weighted center is determined by at most
three weighted points. That is, given an assignment of weights, there always
exist at most three weighted points of $P$ whose weighted center coincides with the
weighted center of the whole weighted points.
Figure~\ref{fig:arrangement_structure}
  illustrates the cases that
  the weighted centers determined by three or two weighted points of $P$.

Thus the number of all possible weighted centers is at most the number of all possible 6-tuples 
$(p_1,p_2,p_3,w_1,w_2,w_3)$ such that $p_i \in P$ and $w_i \in W \cup \{1\}$ for $i=1,2,3$, 
which is $O(k^3n^3)$.
Moreover, this bound is asymptotically
tight by the following lemma.

\begin{lemma}
	\label{lem:lower-bound}
	There are $n$ points in the plane and $k$ weights
	for which the number of all possible weighted centers is $\Omega(k^3n^3)$.
\end{lemma}
\begin{proof} 
  Figure~\ref{fig:tight} illustrates an example in the plane with $n$ points and $k$ weights for
  which the number of all possible weighted centers is
  $\Omega(k^3 n^3)$. Consider two concentric circles such that
  	the smaller circle has a radius at most $1/3$ of the radius of the larger circle.
  	We put $\lfloor n/4\rfloor$ points
  along the small circle $C_2$ evenly and put the remaining points along
  the large circle $C_1$ such that they are clustered into three
  groups, each consisting of at most $\lceil n/4\rceil$ points, as illustrated in the figure.  Any
  two points in the same group are close to each other while any two
  points from two different groups are far from each other.  Let
  $\eps>0$ be a sufficiently small real value and $k$ be a submultiple of $n$ satisfying $k\leq n/4$, 
  and let $W=\{1/2-\eps, 1/2-2\eps,\ldots,1/2-k\eps\}$.

  Every triplet $(p_1,p_2,p_3)$ of points, one point from each group
  on $C_1$, and every triplet $(w_1,w_2,w_3)$ of weights of $W\cup\{1\}$ define
  their weighted center if we assign weight $1$ to all points on $C_1$
  other than the three points. To see this, we first observe that the
  weighted center $c$ of the three points $p_i$ with weights $w_i$ for
  $i=1,2,3$ lies inside $C_2$ by setting $\eps$ to be sufficiently small.
  
  We first show that $d(p_1,c)/w_1> d(q,c)$ for every point $q$ lying on $C_1$ other than
  $p_1,p_2$ and $p_3$. We have $d(p_1,c)>d(q,c)/2$ since $c$ lies
  inside $C_2$ and the radius of $C_2$ is at most $1/3$ of the radius of $C_1$. Therefore, the claim holds
  because $w_1$ is at most $1/2$.

  Now we show that $d(p_1,c)/w_1>d(q,c)/w$ for every point $q$ lying on $C_2$ and 
  every weight $w$ in $W\cup\{1\}$.
  We have $d(p_1,c)/2 > d(q,c)$. By setting $\eps >0$ to be sufficiently small 
  satisfying $1/(1-2k\eps) < 2$,
  $\max_{w,w'\in W} w/w'$ is strictly less than $2$.  Thus,
  $d(p_1,c)/w_1 > d(q,c)/w $ for any point $q$  lying on $C_2$ and any weight
  $w \in W\cup\{1\}$, which proves the claim.

  Moreover, we can choose $\eps$ such that no two weighted centers lie
  on the same point. Therefore, the number of different weighted
  centers is $\Omega(k^3n^3)$. 
\end{proof}
\begin{figure}
	\begin{center}
		\includegraphics[width=0.3\textwidth]{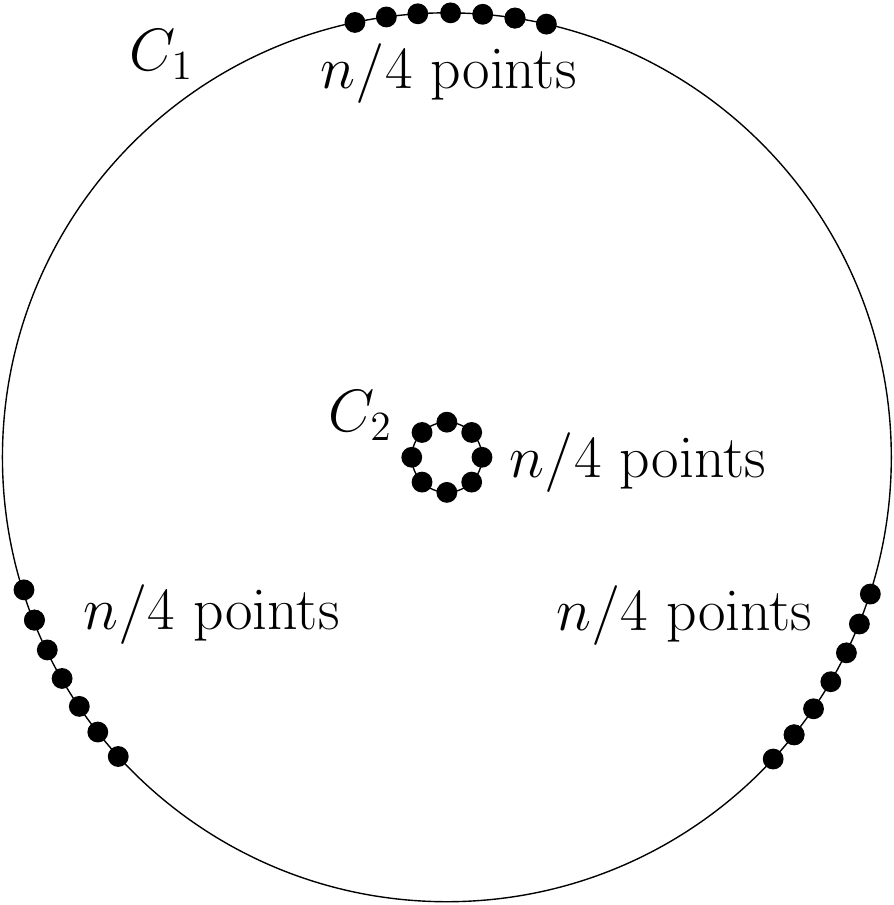}
		\caption {A example of having $\Theta(k^3n^3)$ different weighted centers: $n$ points and the set $W=\{1/2-\eps, 1/2-2\eps,
			\ldots, 1/2-k\eps\}$ of $k$ weights.}
		\label{fig:tight}
	\end{center}
\end{figure}

\subsection{Deciding Feasibility for a Weighted Center of Three Points}
As noted earlier, the weighted center and the covering radius of an
assignment of weights are determined by at most three weighted points.
But not every three weighted points define such a center and its covering radius.
Here we show how to test this for three weighted points efficiently. 

Let $\wone$ be the multiset consisting of the weights of $W$ and $n-k$
  numbers of weight $1$. 
Consider three points from $P$ and an
assignment of three weights from $\wone$ to the points. We first test
whether the three weighted points define a point in the plane at the
same weighted distance from them. This can be done in constant
  time by checking if the three bisecting curves, each curve
  $\{x\in\mathbb{R}^2 \mid d(p,x)/w(p)=d(q,x)/w(q)\}$ defined by a
  pair $(p,q)$ of the three weighted points, have a common point.  If
such a point does not exist, there is no weighted center of the three
weighted points.
Otherwise, there is at most one such point 
as the weighted center of points in the plane is unique.
  We can decide in constant time if the common point of three bisecting
  curves is the weighted center of the three weighted points
  by checking if the point is contained in
  the convex hull of the three weighted points.

Let $c$ denote the weighted center and $r$ denote the covering radius of the 
three weighted points. 
Let $\langle p_1,\ldots,p_{n-3} \rangle$ be
the sequence of the points in $P$, except the three (weighted) points, in
increasing order with respect to the Euclidean distance from $c$.
Let $\langle w_1,\ldots,w_{n-3} \rangle$ be the sequence of the
weights in $\wone$, except the three weights, in increasing order.
The following lemma directly
gives us an $O(n)$-time algorithm to decide whether
the three weighted points determine the weighted center
and the covering radius of an assignment of weights assuming that we have the two sequences.
\begin{lemma}
  \label{lem:validity}
  There exists an assignment $f$ of weights such that $c(f)=c$ and $r(f)=r$ if and only
  if $d(p_i,c)/w_i \leq r$ for $1\leq i\leq n-3$.
\end{lemma}
\begin{proof}
  One direction is straightforward. If $d(p_i,c)/w_i \leq r$ for every
  index $1 \leq i \leq n-3$, we assign weight $w_i$ to point
  $p_i$ so that $f(p_i)=w_i$  for each $i$. 
  Then this assignment $f$, together with the weight assignment for the three points, 
  satisfies that $c(f)=c$ and $r(f)=r$.
  
  Consider the other direction. Suppose that there exists an
  assignment $f$ such that $c(f)$ coincides with $c$ and $r(f)=r$, but 
  $d(p_i,c)/w_i > r$ for some $i$ with $1\leq i \leq n-3$.
  Let $j$ be the smallest index such that $d(p_j,c)/w_j > r$.  
  Then $f(p_j)=w_\ell$ for some index $\ell>j$ because $d(p_j,c)/w_{\ell'}>r$ for
  every $\ell'\leq j$.
  Moreover, there is a point $p_{j'}$ 
  for $j<j'\leq n-3$ such that $f(p_{j'})\leq w_j$. This is because
  the number of weights strictly larger than $w_j$ is less than $n-3-j$.
  Since $d(p_i,c)/f(p_i)\leq r$ for all $1 \leq i \leq n-3$, we
  have $d(p_{j'},c)/f(p_{j'})\leq r$. But this implies $d(p_j,c)/w_j \leq r$ 
  because $d(p_j,c)\leq d(p_{j'},c)$ and $f(p_{j'})\leq w_j$, 
  which contradicts the assumption
  that $d(p_j,c)/w_j > r$.
  Therefore, we have $d(p_i,c)/w_i \leq r$ for all indices $1\leq i
  \leq n-3$ if there exists an assignment $f$ with $c(f)=c$ and $r(f)=r$.
\end{proof}

It is possible that two weighted points $p$ and $q$ determine the weighted center of the whole
weighted points. In this case, the weighted center lies in the line passing through $p$ and $q$.
To handle this case, we consider two points from $P$ and an assignment of two
weights from $\wone$ to the points.
We compute the weighted center and decide whether the two weighted points
determine the weighted center and the covering radius using
Lemma~\ref{lem:validity}.

We need to sort the points in $P$ repeatedly for each weighted center
determined by a combination of
three points (or two points) from $P$ and three weights (or two weights) from $\wone$
while it suffices to sort
the weights in $W$ just once.  Thus, the total running time
for finding the weighted centers of all possible weight assignments is
$O(k^3n^4\log n)$.

Note that this algorithm returns the weighted centers of all possible weight 
assignments. 
Thus
it can be used for problem with optimization criteria 
other than the minimization of the covering radius.  For example, we can find the assignment
$f$ of weights such that the center $c(f)$ is the closest to a given
point in the same time.

\section{Algorithm for Computing an Optimal Assignment of Weights} 
\label{sec:algo1} 
In this section, we present an $O(k^2n^2 \log^3 n)$-time algorithm for
finding an assignment $f$ of weights that minimizes $r(f)$. This
algorithm does not consider all the weighted centers of
  all possible weight assignments. 
Instead, it uses parametric
search due to Megiddo~\cite{parametric}.
To apply this technique, we need to devise a decision algorithm which
is used as a subprocedure of the main algorithm.  

\subsection{Decision Algorithm}
\label{sec:algo1_decision}
Let $r$ be an input of the decision algorithm. The decision algorithm decides
whether there is an assignment $f$ of weights with $r(f) \leq r$.  In
other words, it decides whether there are a point $c$ and an assignment
$f$ of weights such that $d(p,c)/f (p) \leq r$ for all points $p \in
P$.  If this is the case, we call such a point $c$ an 
\emph{$r$-center} 
with respect to  $f$.


\begin{lemma}
  \label{lem:decision}
  For an assignment $f$ of weights with $r(f) \leq r$,
  there is 
  an $r$-center $c$ with respect to $f$ satisfying 
  $d(q,c)/f(q)=r$ for some point $q \in P$.
\end{lemma}
\begin{proof}
  %
  Assume that an assignment $f$ of weights satisfies
  $r(f) =\max_{p\in P} d(p,c(f))/f (p) \leq r$.
  Let $c$ be an $r$-center with respect to $f$. 
  Imagine that we move $c$ along the horizontal line through $c$.
  Since the distance function $d$ is continuous, there always exists
  a point $c'$ on the line such that $\max_{p\in P} d(p,c')/f (p) = r$.
  Clearly, $c'$ is an $r$-center with respect to $f$
  satisfying $d(q,c')/f(q)=r$ for some point $q \in P$.
\end{proof}

By the above lemma, there is a point $p\in P$ and a weight 
  $w \in W\cup\{1\}$
such that the circle centered at $p$ with radius $wr$ contains an $r$-center $c$ 
with respect to an assignment $f$ of weights if $r(f)\leq r$. 
To find such an $r$-center, we do as follows. For
each pair $(p,w)$ of a point $p \in P$ and a weight
$w \in W\cup\{1\}$, we consider the circle $C$ centered at $p$ with
radius $wr$.  
To check whether there exists an $r$-center on $C$, we 
subdivide $C$ into $O(kn)$ pieces such that for each piece $\mu$ on $C$,
every point on $\mu$ is an $r$-center
if and only if a point on $\mu$ is an $r$-center. 
Then we check for each piece if it contains an $r$-center.

\begin{figure}
  \begin{center}
    \includegraphics[width=0.5\textwidth]{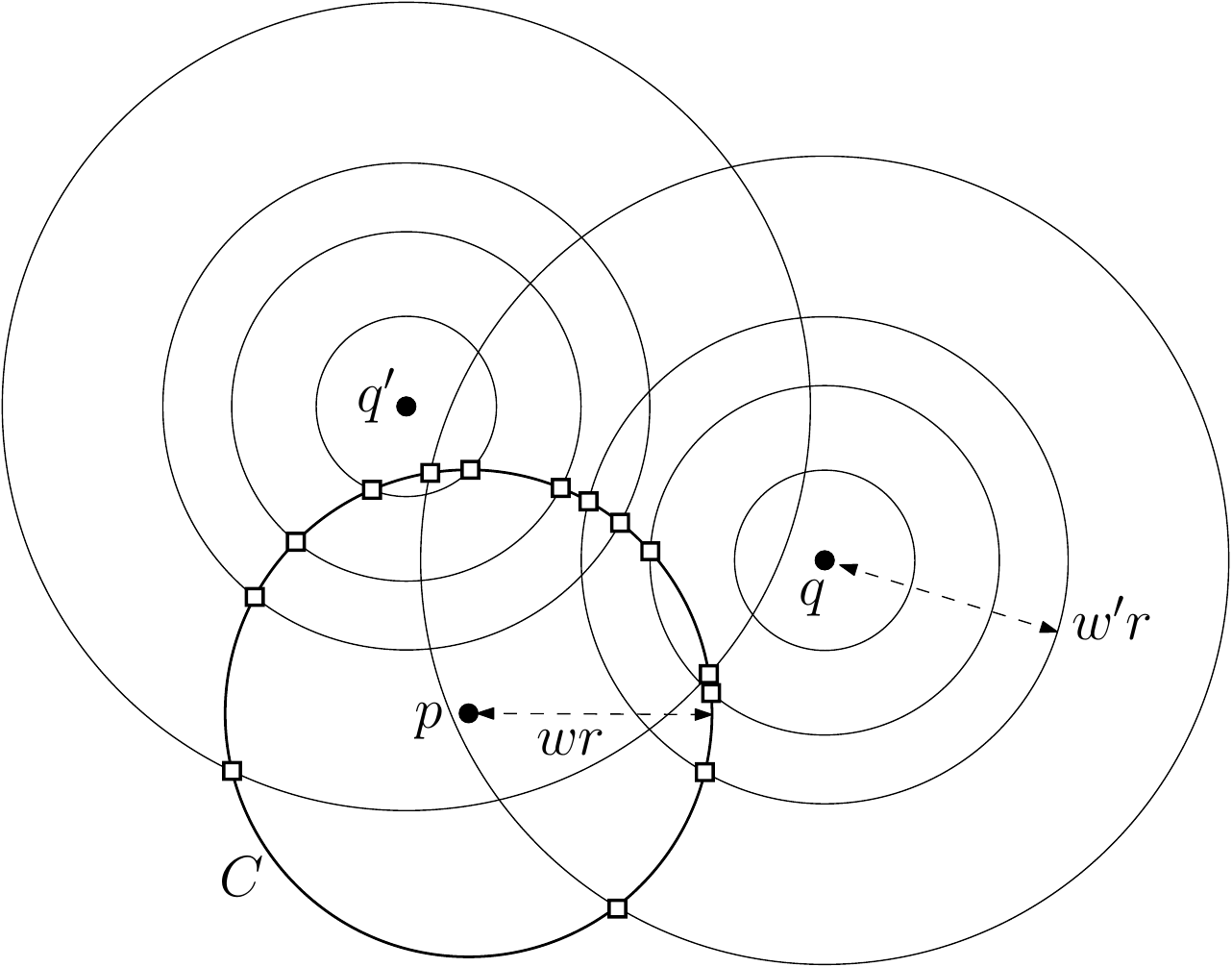}
    \caption {We compute all intersection points of $C$ with other circles,
        each centered at a point $q$ of $P\setminus\{p\}$ with radius $w'r$ for
      a weight $w'$ of $\wnw$. We sort the intersections points along $C$.
    }
    \label{fig:circle_intersection}
  \end{center}
\end{figure}

\subsubsection{Computing the Intervals on the Circles}
Remind that $\wone$ denotes the multiset consisting of the weights of $W$ and $n-k$
  numbers of weight $1$. 
Let $\wnw$ denote $\wone\setminus\{w\}$.
For each point $q \in P\setminus \{p\}$, we compute
concentric circles centered at $q$
with radius
  $w'r$ for each distinct weight $w'$ of $\wnw$. There are at most $k$
concentric circles centered at each point $q\in P\setminus\{p\}$.
We compute the intersections of these circles with $C$ and sort
them along $C$ in $O(kn \log (kn))=O(kn \log n)$ time.
See Figure~\ref{fig:circle_intersection} for an illustration.
In a degenerate case, there can be more than one circle passing through the
same intersection point. In this case, we treat the intersection
points as distinct points lying in the same position and handle them
separately in any order.
In the following, we assume that exactly one circle passes through
one intersection point.

Now, we have $O(kn)$ intersection points sorted along $C$.  These
intersection points subdivide $C$ into $O(kn)$ pieces,
which we call \emph{intervals} on $C$. 
We say an assignment $f$ of weights is \emph{feasible} for an interval
on $C$ if
for a point $c$ in the interval, $d(q,c)/f (q)$ is at most $r$ for
  all points $q \in P\setminus\{p\}$. 
Note that $f$ is feasible for any point in the
interval if $f$ is feasible for a point in the interval.
Therefore, the set of all feasible assignments of weights
is the same for any point $c$ lying in the same interval. 
We can test in $O(n\log n)$ time whether there exists a feasible assignment 
for an interval on $C$ by Lemma~\ref{lem:validity}. 
Specifically, we choose an arbitrary point $c$ in the interval, sort all points in $P\setminus\{p\}$ 
with respect to their distances from $c$, and assign each weight in $\wnw$ to a point in the sorted
list of the points such that a smaller weight is assigned to a point closer to $c$. Then we check if
$d(q,c)/f(q)\leq r$ for every point $q$ of $P\setminus \{p\}$ under the assignment $f$ of weights
we just found.

In this way, we can test the feasibility for every interval on $C$ in 
 $O(kn^2\log n)$ time in total. In the following, we show how to do this 
in $O(kn \log n)$ time in total for all intervals on $C$.

\subsubsection{Test of the Feasibility for Every Interval on \texorpdfstring{$C$}{C}}
Consider the intervals one by one in counterclockwise order along $C$.  Note
that for any two consecutive intervals, there is only one disk
centered at a point in $P\setminus\{p\}$
with radius $rw'$ for some $w' \in \wnw$
that
contains exactly one of the two intervals.

  We first show how to check the existence of a feasible assignment of weights
  for an interval $\mu$ on $C$ in $O(n\log n)$ time.  Then we show how we
  do this for all intervals on $C$ efficiently.
%
  We sort the weights of $W$ and $n-k$ numbers of weight $1$ 
    in the increasing order. The sorted list is denoted by
    $\langle w_1,\ldots,w_{n} \rangle$.
  For each point $q$ in $P$, let
  $\pi(q)$ be the smallest index such that $d(q,c)/w_{\pi(q)} \leq r$,
  where $c$ is a point in $\mu$.
  Thus, $\pi(p)$ is the smallest index of weight $w$ in the sorted list.
  The indices for all points of $P$ can be computed in $O(n\log n)$ time in total.
  Then we sort the points in $P$ in the increasing order with respect
  to $\pi(\cdot)$ in $O(n\log n)$ time and denote the sorted list by
  $\langle q_1,\ldots,q_n\rangle$. 
  Then there exists a feasible assignment of weights for $\mu$ if and only if
  $\pi(q_\ell) \leq \ell$ for all indices $1\leq \ell\leq n$.
  This is because for every point $q\in P$, we have $d(q,c)/w_j \leq r$
  for all indices $j \geq \pi(q)$.
  Since we can check if $\pi(q_\ell) \leq \ell$ for all indices $1\leq \ell\leq n$
  in $O(n)$ time,
  we can decide whether there is a feasible assignment of weights
  for $\mu$ in $O(n\log n)$ time. 

  To check the existence of a feasible assignment of weights for the interval
  $\mu'$ next to $\mu$ in counterclockwise order along $C$,
  we do not need to compute all such indices and compare them again.
  Recall that there is only one disk $D$ centered at a point $p' \in P\setminus\{p\}$
  with radius $rw'$ for some $w'\in \wnw$
  that contains exactly one of $\mu$ and $\mu'$. 
  We observe that $\pi(q)$ for every $q\in P\setminus\{p'\}$
    remains the same for $\mu$ and $\mu'$, because for any point $c'\in\mu'$
    we have $d(q,c')/w_{\pi(q)}\leq r$ and no smaller weight $\bar{w}\in \wnw$
    achieves $d(q,c')/\bar{w} \leq r$.
Now consider $\pi(p')$ for $\mu$. Let $w_\mu:= w_{\pi(p')}$ for $\mu$.
  If $D$ contains $\mu$ but does not contain $\mu'$,
  $D$ has radius $rw_\mu=rw'$ and $d(p',c')>rw_\mu$
  for any point $c'\in\mu'$. Thus, for $\mu'$, $\pi(p')$
  increases to the smallest index such that $d(p',c)/w_{\pi(p')}\leq r$.
  More specific, $w_{\pi(p')}$ for $\mu'$ is the smallest weight larger than
  $w_\mu$ among weights in $\wnw$.
  If $D$ contains $\mu'$ but does not
  contain $\mu$, $\pi(p')$ decreases to the smallest index such that
  $d(p',c)/w_{\pi(p')}\leq r$ for $\mu'$. Similarly, $w_{\pi(p')}=w'$ for $\mu'$
  is the largest weight smaller than
  $w_\mu$ among the weights in $\wnw$.

  To use this property, we maintain $2t$ pointers
  $U_1,\ldots,U_{t}$ and $L_1,\ldots,L_{t}$, where $t$ with $t\leq k$
  is the number of distinct weights in $\wnw$. For an index $i$ with
  $1\leq i\leq t$, the pointer $U_i$ points to $q_\ell$ with the
  largest index $\ell$ such that $\pi(q_\ell)$ is
  equal to the $i$th largest weight in $\wnw$.  
  Similarly, for an index $i$ with $1\leq i\leq t$, the
  pointer $L_i$ points to $q_\ell$ with the smallest index $\ell$
  such that $\pi(q_\ell)$ is equal to the $i$th largest weight in $\wnw$.
    
  Note that we already know whether $\pi(p')$ increases or decreases
  when we move from $\mu$ to $\mu'$.  Here, we have to update not
  only $\pi(p')$ but also the pointers.  Moreover, we have to reorder
  $\langle q_1,\ldots,q_{n} \rangle$ in the increasing order with
  respect to $\pi(\cdot)$.
  
  We show how to do this for the case that $\pi(p')$ increases.
  The other case can be handled analogously.  We first find the point
  $p'_u$ that the pointer $U_{\pi(p')}$ points to.  Then we swap the
  positions for $p'$ and $p'_u$ on the sequence
  $\langle q_1,\ldots,q_{n} \rangle$ and let $U_{\pi(p')}$ point to
  $p'$.  This does not violate the property that
  $\pi(q_\ell)\leq \pi(q_{\ell+1})$ for all indices $1\leq \ell < n$
  since $\pi(p')=\pi(p'_u)$.  Then we increase $\pi(p')$ and
  update the pointers accordingly.	
  To check the existence of a feasible assignment of weights for $\mu'$, it suffices to
  check $\pi(q_\ell)\leq \ell$ for $q_\ell =p', p'_u$ only
  because there is no change to $\pi(q)$ for the other
  points $q$. This can be done in constant time, and thus the update for
  the next interval can be done in constant time.  Since there are $O(kn)$ intervals on $C$,
  the update and existence test can be done in $O(kn)$ time in total for the intervals.
  Therefore, the running time of the procedure is dominated by the time for sorting the
  intersection points along $C$ and computing the intervals, which is
  $O(kn \log n)$.

\begin{lemma}
	\label{lem:interval}
	We can decide whether or not there is a feasible assignment of weights 
	for every interval on $C$ in $O(kn\log n)$ time in total.
\end{lemma}

Since there are $O(kn)$ point-weight pairs,
we have the following lemma.
\begin{lemma}
  Given a real value $r>0$, we can decide in $O(k^2n^2\log n)$ time
  using $O(kn)$ space whether or not there exists an assignment $f$ of weights
  with $r(f)\leq r$.
\end{lemma}

\subsection{Overall Algorithm}
\label{sec:algo1_overall}
In this section, we present an algorithm for computing an optimal assignment of weights
in $O(k^2n^2\log^3 n)$ time using
$O(kn)$ space.  
To obtain an optimal solution, we apply parametric search~\cite{parametric} 
using the decision algorithm in
Section~\ref{sec:algo1_decision}.  Let $r^*$ be the minimum of $r(f)$
over all possible assignments $f$ of weights.
We use $A(r)$ to denote the arrangement of the circles $C_{p,w}(r)$ for all
point-weight pairs $(p,w)$, where $C_{p,w}(r)$ is the circle centered
at $p$ with radius $rw$.  Let $\mathcal{C}(r)$ be the set of the 
circles $C_{p,w}(r)$ for all point-weight pairs $(p,w)$.
Here, $r > 0$ is a variable.  
Notice that the combinatorial complexity of $A(r)$ is $\Theta(k^2n^2)$ in the worst case. Thus 
we cannot maintain the whole structure of $A(r)$ explicitly using $O(kn)$ space.

\subsubsection{Combinatorial Structure of the Arrangement}
Imagine that $r$ increases from $0$.
  Then the combinatorial structure of $A(r)$ changes at certain $r$ values.
To be specific, we observe that 
three circles 
  of $\mathcal{C}(r)$ 
  meet at a point or two circles of $\mathcal{C}(r)$
    become tangent to each other
  when the combinatorial structure of $A(r)$ changes.
 We call a value of $r$
where the combinatorial structure of $A(r)$ changes
an \emph{$r$-value}.
For any three point-weight pairs, 
there are at most two $r$-values for which the circles $C_{p,w}(r)$ of the three pairs $(p,w)$
have a common intersection because the trajectory of the intersections between
two increasing circles forms a hyperbolic curve and two hyperbolic curves cross
at most twice.  
Also, for any two point-weight pairs,
there are at most two $r$-values for which the circles $C_{p,w}(r)$
of the pairs $(p,w)$ are tangent to each other.
Thus, the combinatorial structure of $A(r)$ changes
$O(k^3n^3)$ times, and there are $O(k^3n^3)$ $r$-values.

Given a real value $r > 0$, we first introduce a simple way to compute
the arrangement $A(r)$. 
For a point $p \in P$ and a weight $w \in
W\cup\{1\}$, we compute the intersections of $C_{p,w}(r)$ and
$C_{p',w'}(r)$ for all points $p'$ in $P\setminus\{p\}$ and all
distinct weights $w'$ of $\wnw$.
Each circle $C_{p',w'}(r)$ intersects $C_{p,w}(r)$ at most twice for any fixed $r > 0$.
Let $\mathcal{I}_{p,w}(r)$ be the set of all such $O(kn)$ intersection points.
We sort the points in $\mathcal{I}_{p,w}(r)$
along $C_{p,w}(r)$ in counterclockwise order.
Once this is done for all points $p\in P$ and all distinct weights 
$w$ of $W\cup\{1\}$, we can construct $A(r)$.

As $r$ increases, the combinatorial structure of $A(r)$ changes
only when the sorted list of the points
in $\mathcal{I}_{p,w}(r)$ along $C_{p,w}(r)$ for a point-weight pair $(p,w)$
changes. Clearly, the sorted list changes if two points in $\mathcal{I}_{p,w}(r)$ cross
each other, a point appears on $\mathcal{I}_{p,w}(r)$, or a point  disappears
from in $\mathcal{I}_{p,w}(r)$. Note that an intersection point $x$ appears on $\mathcal{I}_{p,w}(r)$ or disappears from $\mathcal{I}_{p,w}(r)$ 
if $C_{p',w'}(r)$ becomes tangent to $C_{p,w}(r)$ for a point-weight pair $(p',w')$.

These $r$-values partition the real value space $\mathbb{R}$
into intervals such that 
for any value $r$ in the same interval
the combinatorial structure of $A(r)$ remains the same.
We search for the interval where $r^*$
lies using the decision algorithm in Section~\ref{sec:algo1_decision}.
There are $\Theta(k^3n^3)$ such $r$-values, 
but we do not consider all of them.  Instead, as we will see later, 
we do the search 
in $O(\log n)$ iterations.  
In each iteration, we consider
$O(k^2n^2)$ $r$-values and reduce the search space (intervals).  
In each iteration, we apply the decision algorithm $O(\log n)$ times,
which leads to the running time of $O(k^2n^2\log^3 n)$.  

\subsubsection{Computing the Combinatorial Structure of \texorpdfstring{$A(r^*)$}{A(r*)}}
  In this subsection, we compute the combinatorial structure of 
  $A(r^*)$ without computing $r^*$ explicitly.   
  We first compute the intersection points in $\mathcal{I}_{p,w}(r^*)$ for each point-weight
  pair $(p,w)$, and then sort them along $C_{p,w}(r^*)$.
  For both procedures, we need the following technical lemma.
 \begin{lemma}\label{lem:median-of-median}
  	Let $\mathcal{L}_i$ be a set of $M$ real numbers for an index $i$ with $1\leq i\leq N$.
	We can find the interval containing $r^*$ among the intervals on $\mathbb{R}$ induced by the numbers of 
	$\cup_{i=1}^N \mathcal{L}_i$
	in $O((k^2n^2\log n+NM)\log (N+M))$ time using $O(kn+M)$ space assuming that we can 
	obtain the numbers in $\mathcal{L}_i$ in $O(M)$ time.
\end{lemma}
\begin{proof}
	If we are allowed to use $O(kn+NM)$ space, we can compute the intervals on $\mathbb{R}$ induced
	by the numbers in $\cup_{i=1}^N \mathcal{L}_i$ explicitly, and then apply binary search
	on the intervals using the decision algorithm in Section~\ref{sec:algo1_decision}.
	This takes $O((k^2n^2\log n+NM)\log (N+M))$ time.
	
	We can improve the space complexity to $O(kn+M)$ as follows. Basically, we apply binary search
	on the intervals on $\mathbb{R}$ induced by the numbers in each set $\mathcal{L}_i$.
	Initially, $\mathbb{R}$ itself is the search space for every set $\mathcal{L}_i$. 
	For each $i$ with $1\leq i\leq N$, the search space for $\mathcal{L}_i$  gets smaller in over iterations. 
	After $O(\log (NM))$ iterations, the search space for $\mathcal{L}_i$ is an interval
	of $\mathbb{R}$ induced by the real numbers in $\mathcal{L}_i$. 
	Then the intersection of all search spaces is the interval we want to find.
	
	In each iteration, we compute the median $a_i$ of the real numbers in $\mathcal{L}_i$ 
	in the current search space for each set $\mathcal{L}_i$ in $O(NM)$ time in total. 
	We also compute the number $b_i$ of the real numbers of $\mathcal{L}_i$ 
	in the current search space in the same time.
	We assign weight $b_i/b$ to $a_i$ for each set $\mathcal{L}_i$, where $b$ is the sum of $b_i$
	for all indices $i$ with $1\leq i\leq N$. Then 
	we compute the \emph{weighted median} of all medians $a_i$'s of $\mathcal{L}_i$'s. In specific,
	we compute $a_i$ with smallest index $i$ 
	such that the sum of $a_jb_j$ for all sets $\mathcal{L}_j$ with $a_i\geq a_j$
	is at most $1/2$.
	After computing $a_i$ and $b_i$ for every set $\mathcal{L}_i$, we can find the weighted median $a$ of them
	in $O(N)$ time. Let $a'$ be the smallest median larger than $a$
	among the medians for all $\mathcal{L}_i$.  We apply the decision algorithm to $a$ (and $a'$)
	to determine
	if $r^*\leq a$, $a < r^* \leq a'$, or $a' < r^*$.
	If $r^* \leq a$, we halve the search space for $\mathcal{L}_i$ such that $a_i \geq a$.
	If $a < r^* \leq a'$, we halve the search space for $\mathcal{L}_i$ such that $a_i=a$. 
	Otherwise, we halve the search space for $\mathcal{L}_i$ such that $a_i\leq a$.
	Each iteration takes $O(k^2n^2\log n + NM)$ time using $O(kn+M)$ space,
	and reduces the search space for some sets.
	
	The search space for $\mathcal{L}_i$ becomes an interval of $\mathbb{R}$ induced 
	by the real numbers in the set in $O(\log (N+M))$ iterations. This is simply because 
	the sum of $b_i$'s reduces by a constant factor for each iteration. 
	Therefore, we can compute the interval containing $r^*$ among the intervals 
	on $\mathbb{R}$ induced by the real numbers of $\cup_{i=1}^N \mathcal{L}_i$
	in $O((k^2n^2\log n+NM)\log (N+M))$ time using $O(kn+M)$ space.	
\end{proof}
  
  \paragraph{Computing the Intersection Points in $\mathcal{I}_{p,w}(r^*)$.}
  To compute the intersection points in $\mathcal{I}_{p,w}(r^*)$, we compute 
  every $r$-value for which a point appears on or disappears from $\mathcal{I}_{p,w}(r^*)$.
  Recall that a circle $C_{p',w'}(r)$ is tangent to $C_{p,w}(r)$ at such a $r$-value.
  There are $O(k^2n^2)$ such $r$-values for all point-weight pairs $(p,w)$, and we can compute them
  in $O(k^2n^2)$ time. By applying binary search on such $r$-values
  together with the decision algorithm in Section~\ref{sec:algo1_decision}, we find the interval containing $r^*$ among the intervals
  on $\mathbb{R}$ induced by such $r$-values. This takes $O(k^2n^2\log^2 n)$ time with $O(k^2n^2)$ space.
  
  But we can improve the space complexity to $O(kn)$ without increasing the time complexity
  using Lemma~\ref{lem:median-of-median}. Here, we have a set of $r$-values for each point-weight pair $(p,w)$. We can compute the $r$-values in each set in $O(kn)$ time. We want to compute
  an interval on $\mathbb{R}$ induced by all $r$-values for all point-weight pairs.
  We apply Lemma~\ref{lem:median-of-median} to the sets of $r$-values. Then we can 
  compute the interval containing $r^*$ among the intervals
  on $\mathbb{R}$ induced by such $r$-values in $O(k^2n^2\log^2 n)$ time with $O(kn)$ space.
  
  Since for any point-weight pair, no point appears on or disappears from $\mathcal{I}_{p,w}(r)$ in this interval, we can obtain the intersection points in $\mathcal{I}_{p,w}(r^*)$ by computing
  the intersection points in $\mathcal{I}_{p,w}(r)$ for any value
  $r$ in the interval.
  
  \paragraph{Sorting the Intersection Points in $\mathcal{I}_{p,w}(r^*)$ along $C_{p,w}(r^*)$ in Clockwise Order.}
  Since $C_{p,w}(r)$ is a circle for any value $r>0$, we use an arbitrary point
  in $\mathcal{I}_{p,w}(r^*)$,
  denoted by $p_0(r)$, as the reference point of the sorted list.
  That is, the points in $\mathcal{I}_{p,w}(r)$ are sorted along $C_{p,w}(r)$,
  say in counterclockwise order, from $p_0(r)$.
  To make the description clear, we first handle the $r$-values for which $u(r)$ crosses
  $p_0(r)$ for every intersection point $u(r)$ in $\mathcal{I}_{p,w}(r^*)$.
  We compute the $r$-values for each point-weight pair in $O(kn)$ time, and then apply
  binary search on the $r$-values for all point-weight pairs as we did before 
  in $O(k^2n^2\log^2 n)$ time with $O(kn)$ space. We obtain an interval on $\mathbb{R}$ such that
  no intersection point $u(r)$ crosses $p_0(r)$ in this interval for any point-weight pair.
  
  Suppose that we want to compute the relative positions of two
  points $u_1(r^*)$ and $u_2(r^*)$ in the sorted list of the points of
  $\mathcal{I}_{p,w}(r^*)$ along $C_{p,w}(r^*)$. 
  As $r$ increases,
  both $u_1(r)$ and $u_2(r)$ move along $C_{p,w}(r)$ and cross each other
  at most twice. 
  Moreover, they cross each other 
  when $C_{p,w}(r)$ and the two circles defining $u_1(r)$ and $u_2(r)$ meet at a point.
  This occurs at most twice. Let $r_1$ and $r_2$ be the two $r$-values with
  $r_1 \leq r_2$ at which $u_1(r)$ and $u_2(r)$ cross each other.
  Then we decide whether $r^* \leq r_1$, $r_1\leq r^*
  \leq r_2$, or $r_2 \leq r^*$ using the decision algorithm in
  Section~\ref{sec:algo1_decision}.  With this, we can decide the
  relative position of $u_1(r^*)$ and $u_2(r^*)$ with respect to $p_0(r^*)$
    in the sorted list of the points in $\mathcal{I}_{p,w}(r^*)$ without computing $r^*$
  explicitly.
   Thus, we can compare two points in
   $\mathcal{I}_{p,w}(r^*)$ in $O(k^2n^2\log n)$ time.  By applying this comparison
   $O(kn\log n)$ times, we can sort the points in $\mathcal{I}_{p,w}(r^*)$
  along $C_{p,w}(r^*)$ without knowing $r^*$ in
  $O(k^3n^3\log^2 n)$ time.
  Since there are $O(kn)$ point-weight pairs, the total
  running time is $O(k^4n^4\log^2 n)$.
	
  Here, we apply the decision algorithm in
  Section~\ref{sec:algo1_decision} twice for each comparison, once with $r=r_1$
  and once with $r=r_2$. We reduce the
  running time of the overall algorithm by reducing the number of executions of 
  the decision  algorithm.
  Suppose that we want to do $m$ comparisons which are independent to
  each other.  As we did in the previous procedure, we
  compute at most two $r$-values from each comparison where
  two points cross each other. 
  Then we have at most $2m$ $r$-values. We sort them and apply binary search to compute the smallest
  interval containing $r^*$.
  After applying the decision algorithm $O(\log m)$ times, we can
  complete the $m$ comparisons.  
	
  In our problem, we use Cole's parallel algorithm 
  to sort $m$ elements in $O(\log m)$ time
  using $O(m)$ processors~\cite{parallel-sorting}. Note that comparisons performed in different
  processors are independent to each other.
  In each iteration, we compare $O(kn)$ elements for each point-weight pair.
  To do this, we consider a set of $O(kn)$ $r$-values for each point-weight pair, and apply Lemma~\ref{lem:median-of-median} to the sets for all point-weight pairs.
  Then we can obtain an interval containing $r^*$ among the intervals on $\mathbb{R}$ induced by the $r$-values for each
  point-weight pair in $O(k^2n^2\log^2 n)$. 
  In other words, we can complete the $O(kn)$ comparisons for every point-weight pair. 
  After $O(\log n)$ iterations, we can obtain an interval containing $r^*$ among the intervals on $\mathbb{R}$ induced by all $r$-values.
  This takes $O(k^2n^2\log^3 n)$ time using $O(kn)$ space.
  The combinatorial structure of $A(r)$ remains the same in this interval. Thus 
  we obtain the combinatorial structure of $A(r^*)$ by computing the combinatorial structure of $A(r)$. 

\begin{lemma}
	\label{lem:parallel}
	The combinatorial structure of $A(r^*)$ can be computed in
	$O(k^2n^2\log^3 n)$ time using $O(kn)$ space.
\end{lemma}

\subsubsection{Finding an Optimal Assignment of Weights}
We have the combinatorial structure of $A(r^*)$ while $r^*$
is not known yet. Using the combinatorial structure, we show how to compute 
an optimal assignment of weights and its covering radius $r^*$ 
in $O(k^2n^2 \log^2n)$ time in this subsection.

  We say that three circles,
  each defined by a point-weight pair,
  \emph{determine an $r$-value $r'$} if they intersect at one
  point for $r=r'$. 
  We already showed that there are at most
    three circles $C_1^*, C_2^*,$ and $C_3^*$,
    each defined by a point-weight pair for $r$-value $r^*$,
  that determine $r^*$.  Thus, in the combinatorial structure
    of $A(r^*)$, two intersection points on one of the three circles,
    say $C_1^*$, made by the other two are consecutive along $C_1^*$.

  Let $R$ be the set of all $r$-values determined by three
    circles in the combinatorial structure of $A(r^*)$
  such that 
    two intersection points on one of the three circles, denoted by
    $C$, made by the other two are consecutive along $C$.
  Since there are $O(k^2n^2)$ edges in the arrangement $A(r^*)$, there
  are the same number of such $r$-values.
	
  Imagine that we sort the $r$-values of $R$, apply binary search on the sorted list
  using the decision algorithm in Section~\ref{sec:algo1_decision}, and
  find the smallest $r$-value of $R$ to which the decision algorithm
  returns ``yes''.  Then the smallest $r$-value is $r^*$.
  This takes $O(k^2n^2\log^2 n)$ time as the $r$-values of $R$ can be sorted
  in $O(k^2n^2\log n)$ time and the decision algorithm is called $O(\log n)$
  times. 
  
  We do this without computing the whole arrangement $A(r^*)$.
  We consider a point-weight pair $(p,w)$ first. We compute the edges and the vertices lying
on $C_{p,w}(r^*)$.  This takes $O(kn\log n)$ time because we already
have the interval $\mu$ containing $r^*$ such that
	the combinatorial structure of $A(r)$
	remains the same for any $r$ in $\mu$. 
Then we apply the algorithm in Lemma~\ref{lem:find} on the edges and
vertices lying on $C_{p,w}(r^*)$.  The algorithm returns the minimum
	$r$-value in $\mu$ 
to which the decision algorithm returns ``yes''.
We do this for all point-weight pairs $(p,w)$.  Then we
have $O(kn)$ $r$-values one of which is exactly $r^*$.  We again apply
binary search on these $r$-values to find the minimum value over
such $r$-values to which the decision algorithm returns ``yes''.
Clearly, the minimum value is exactly $r^*$. This algorithm takes $O(k^2n^2\log^2 n)$ time using $O(kn)$ space.

\begin{lemma}
	\label{lem:find}
	Given the combinatorial structure of $A(r^*)$, an optimal assignment of weights
	and its covering radius $r^*$ can be computed in $O(k^2n^2\log^2
	n)$ time using $O(kn)$ space.
\end{lemma}

\begin{theorem}
	Given a set $P$ of $n$ points in the plane and a multiset $W$ of $k$
	weights with $k \leq n$, an assignment $f$ of weights that minimizes $r(f)$ can
	be found in $O(k^2n^2 \log^3 n)$ time using $O(kn)$ space.
\end{theorem}

\section{Faster Algorithm for a Small Set of Weights}
We can improve the algorithm in Section~\ref{sec:algo1} for the case
that $k$ is sufficiently small compared to $n$ and the input weights are
at most $1$. Specifically, we can compute an optimal assignment of weights
in $O(k^5n \log^3 k + kn\log^3 n)$ time using $O(kn)$ space. 
Let $f^*$ be an optimal assignment of weights, that is, $r(f^*)=r^*$. 
A point $p$ in $P$ is called a \emph{determinator} of $f^*$
if $d(p,c(f^*))/f^*(p)=r^*$.
We already
observed that there are two or three determinators for an optimal
assignment of weights. 
Moreover, if there exist exactly two determinators, then the two
determinators and $c(f^*)$ are collinear.

The algorithm in this section is based on the observation that if
$f^*(p)=1$ for a determinator $p$, then $d(p,c(f^*))\geq d(q,c(f^*))$ 
for any point $q$ in $P$. 
The following lemma provides a more general observation. 
\begin{lemma}
  \label{lem:determinator}
  Let $f^*$ be an optimal assignment of weights with minimum number
  of determinators.
  If a determinator $p$ is the $i$th closest point of $P$ from $c(f^*)$, 
  it is assigned the $i$th weight value in the sorted list of the weights
    of $\wone$ in increasing order, where $\wone$ is the multiset consisting of
    the weights of $W$ and $n-k$ numbers of weight $1$.
\end{lemma}
\begin{proof}
  Consider a determinator $p$ that is the $i$th closest point of $P$
  from $c^*=c(f^*)$.  For any point $q \in P$ with
  $d(q,c^*) \geq d(p,c^*)$, we have $f^*(q) \geq f^*(p)$.  This is
  because $d(p,c^*)/f^*(p)=r^* \geq d(q,c^*)/f^*(q)$.
	
  Consider a point $q \in P\setminus\{p\}$ with $d(q,c^*) < d(p,c^*)$. If
  $f^*(q) > f^*(p)$, then we have $d(p, c^*)/f^*(q)< r^*$ and
  $d(q, c^*)/f^*(p)<r^*$. This implies that we can remove one
  determinator, $p$, of $f^*$ without increasing the covering radius
  by swapping the weights assigned to $p$ and $q$, which is a
  contradiction.  Thus, we have $f^*(q) \leq f^*(p)$.
	
  Therefore, there are $i-1$ points $q\in P\setminus\{p\}$ satisfying
  $f^*(q)\leq f^*(p)$ and $n-i$ points $q'\in P\setminus\{p\}$
  satisfying $f^*(q')\geq f^*(p)$.  This implies that
    $f^*(p)$ is the $i$th weight value in the sorted list of the weights of $\wone$
    in increasing order.
\end{proof}

We will see that the number of candidates for 
a determinator with its weight we consider can be reduced by Lemma~\ref{lem:determinator}.
Recall that the algorithm in Section~\ref{sec:algo1} considers
each of $O(kn)$ point-weight pairs as a determinator and its weight.
Also, we use the following corollary, which is a paraphrase of 
	Lemma~\ref{lem:validity}.
\begin{corollary}
  There exists an optimal assignment $f^*$ of weights that assigns the $i$th closest point
  of $P$ from $c(f^*)$ the $i$th weight value in the sorted list of the weights
  of $\wone$ in increasing order, where
  $\wone$ is the multiset consisting of the weights of $W$ and
  $n-k$ numbers of weight $1$.
\end{corollary}

We consider two possible cases of optimal weight assignments with respect to the determinators
and handle them in different ways. The first case is that every
determinator is assigned a weight smaller than 1. 
The second case is that a determinator is assigned weight $1$.


\subsection{Every Determinator is Assigned a Weight Smaller than 1.}  
\label{sec:every-determinator-smaller}
By Lemma~\ref{lem:determinator}, a determinator is one of the $k$ closest
points from $c(f^*)$.  This is related to the \emph{order-$k$
Voronoi diagram}.  The order-$k$ Voronoi diagram is a generalization of
the standard Voronoi diagram. Given sites in the plane, it partitions the plane into regions
such that every point in the same region has the same $k$ closest sites.
The complexity of the order-$k$ Voronoi diagram of $n$ point sites is
$O(kn)$~\cite{higherVDcomplexity}.  There are a number of algorithms
to compute the order-$k$ Voronoi diagram with different running
times~\cite{higherVD1, higherVD2}.  We compute the
order-$k$ Voronoi diagram of $P$ using the algorithm
in~\cite{higherVD2}, which runs in $O(n\log n+nk2^{\alpha\log^* k}) \leq
O(n\log n +kn\log k)$ time using $O(kn)$ space, where $\alpha$ is a constant. 
We also compute the farthest-point Voronoi diagram of $P$,
also known as the order-$(n-1)$
Voronoi diagram of $P$, in $O(n\log n)$ time.

In terms of the order-$k$ Voronoi diagram,
  Lemma~\ref{lem:determinator} can be interpreted as follows. The
  points of $P$ are the sites of the Voronoi diagram and the $k$ sites
  corresponding to the cell $V^*$ containing $c(f^*)$ are the closest
  $k$ sites from $c(f^*)$.  Note that the weights of $W$ are the
  smallest $k$ weights among the weights of $W$ and $n-k$ numbers of
  weight $1$.  By Lemma~\ref{lem:validity}, all the $k$ sites
  corresponding to $V^*$ must be assigned the weights of $W$ and all
  the other sites must be assigned weight $1$.  Since every
  determinator is assigned a weight smaller than $1$, it must be
  one of the $k$ sites corresponding to $V^*$.

  To use this observation, we consider each cell of the order-$k$ Voronoi diagram
  of $P$.  For each cell, we assign the $k$ weights of $W$
  to the $k$ sites corresponding to the cell, each weight to a distinct site
  by applying the algorithm in Section~\ref{sec:algo1} to the $k$ sites.
  This takes $O(k^4 \log^3 k)$ time. Let $c$ and $r$ be the weighted
  center and the covering radius of this assignment of weights with respect to the $k$ sites.
  We assign weight $1$ to all the other $n-k$ sites.
  Let $f$ be the resulting assignment of weights.

Then we check whether $c$ is $c(f)$ and $r$ is $r(f)$.
  By Lemma~\ref{lem:validity}, it suffices to check if the farthest point
  from $c$, which is assigned weight $1$, lies at distance at most $r$ from $c$.
This can be done in $O(\log n)$ time by finding the farthest point of $P$ from $c$
using the farthest-point Voronoi diagram of $P$.
In total, this takes $O(k^5n \log^3 k +kn \log n)$ time.

\subsection{A Determinator is Assigned Weight 1.}  
We first apply the procedure that deals with
the first case in Section~\ref{sec:every-determinator-smaller}. 
Let $r_U$ be the minimum radius of the results of
this procedure.  To handle the case that a determinator is assigned weight 1, 
we present a decision algorithm that returns ``yes'' for an input $r$ with $0 < r < r_U$ if and only if there
exist an assignment $f$ of weights and a point $c \in \mathbb{R}^2$
such that $d(p,c)/f (p) \leq r$ for all points $p \in P$ and $d(q,c) =
r$ for some point in $q \in P$.  
In other words, the decision algorithm returns ``yes'' if and only if
there is an assignment of weights with covering radius $r$ one of whose determinators is assigned
weight 1.
For a radius $r$, we call such a point $c$ a \emph{center} with radius $r$.

The following lemma enables us to apply parametric search.
\begin{lemma}
	\label{lem:monotone}
  If a determinator of an optimal  weight assignment is assigned weight $1$, then the
  decision problem for any input $r$ with $r^* \leq r < r_U$ returns 
  ``yes''.
\end{lemma}
\begin{proof}
	Assume to the contrary that the decision problem for
	an input $r$ with $r^* \leq r < r_U$ returns ``no''. 
	%
	By Lemma~\ref{lem:decision}, there is an assignment $f$ of weights with covering radius $r$.
	However, since the answer for the decision problem is ``no'', every determinator of 
	$f$ is assigned a weight smaller than 1.
	Then such an assignment should have been dealt by the procedure for the first case
	in Section~\ref{sec:every-determinator-smaller}, which contradicts 
          the assumption that $r < r_U$. 
	Therefore, the lemma holds.
\end{proof}

\subparagraph{Decision Algorithm.}  Given a covering
radius $r$, the decision algorithm first computes the intersection $I$
of the disks $D(p,r)$ for all points $p \in P$, where $D(p,r)$ is the
disk centered at $p$ with radius $r$. If the answer for the decision
problem is ``yes'', 
there is a center corresponding to covering radius $r$ that
lies on the boundary of $I$ by the definition.

Thus the decision algorithm searches the boundary of $I$ and checks
whether there exists a center on the boundary of $I$. The boundary of $I$
consists of circular arcs. We call an endpoint of a maximal circular arc on the boundary of $I$
a \emph{breakpoint}. Here, we follow
the framework of the algorithm in Section~\ref{sec:algo1}.  That is,
we consider $O(kn)$ circles $C_{p,w}(r)$ for all points $p \in P$ and
all distinct weights $w$ of $W\setminus\{1\}$, where $C_{p,w}(r)$ is the
circle centered at $p$ with radius $rw$.  We compute
  $O(kn)$ intersection points of the circles with the boundary of $I$.
We sort them together with the breakpoints of $I$ 
along the boundary of $I$ in $O(kn\log n)$ time. Then we
apply the procedure in the proof of Lemma~\ref{lem:interval}, which
checks whether there exists a center corresponding to 
  covering radius $r$ lying on the boundary of $I$ in $O(kn)$ time.
The decision algorithm takes $O(kn\log n)$ time.

\subparagraph{Overall Algorithm.}
As we did in the decision algorithm, we first compute the intersection
$I(r^*)$ of the disks $D(p,r^*)$ for all points $p \in P$. Here, we
are not given $r^*$.  Instead of computing the intersection
explicitly, we compute its combinatorial structure.

\begin{lemma}
  The combinatorial structure of the intersection of the disks
  $D(p,r^*)$ for all points $p \in P$ can be computed in $O(n\log n+
  T(n)\log n)$ time, where $T(n)$ is the running time of the decision
  algorithm.
\end{lemma}
\begin{proof}
  As $r$ increases from $0$ to $r_U$, the combinatorial
  structure of the intersection of the disks $D(p,r)$ may change.  We
  have two types of events where the combinatorial structure changes: 
  an arc of a disk starts to appear in the structure 
  or an existing arc of a disk disappears from the structure.  
  At both types of events, such an arc becomes a point which is a degenerate
  arc. Moreover, this point is a vertex of the farthest-point Voronoi diagram of $P$.
	
  To use this fact, we compute the farthest-point Voronoi diagram of
  $P$ in $O(n\log n)$ time.  Then for each vertex $v$ of
  the diagram, we compute the Euclidean distance between $v$ and its
  farthest point in $P$.  There are $O(n)$ distances, and we sort them 
  in increasing order.	
  Then we apply binary search on the distances using the decision
  algorithm to find the smallest interval containing $r^*$ whose endpoints are the distances we have. 
  This can be done in $O(T(n)\log n)$ time.
	
  Therefore, for any value of $r$ in the interval, the combinatorial structure of
  the intersection of the disks remains the same. 
\end{proof}

Now we have the combinatorial structure of the intersection
$I(r^*)$. As we did in the algorithm of Section~\ref{sec:algo1}, 
we sort the intersections of $O(kn)$ circles $C_{p,w}(r^*)$ for 
all point $p \in P$ and all distinct weights
  $w$ of $W$ with the boundary of $I(r^*)$
without explicitly computing
$r^*$. This can be done in $O(kn\log n + T(n)\log^2 n)$ time, where $T(n)=O(kn\log n)$
is the running time of the decision algorithm, in a way similar to
Lemma~\ref{lem:parallel}.  Then we find an optimal solution in a way
similar to Lemma~\ref{lem:find} in $O(kn)$ time if it belongs to the second
case.  In total, the second case can be dealt in $O(kn\log^3 n)$ time using
$O(kn)$ space.

Combining the three cases, we have the following theorem.

\begin{theorem}
  Given a set $P$ of $n$ points in the plane and a multiset $W$ of $k$
  weights smaller than or equal to $1$ with $k \leq n$, we can compute an assignment 
  $f$ of weights  that minimizes
  $r(f)$ in $O(k^5n \log^3 k + kn\log^3 n)$ time using
  $O(kn)$ space.
\end{theorem}
\section{Concluding Remarks}
We would like to mention that the approaches presented in this paper also work
under any convex distance function, including the $L_p$ metric for $p\geq 1$. 
For the $L_1$ or the $L_\infty$ metric, there can be more than one
optimal weighted center though. 
The weight assignment problem can also be considered in higher dimensions.
A future work to consider is to devise efficient algorithms for the problem 
in higher dimensions under various distance functions.

\bibliography{paper}{}

\begin{thebibliography}{10}

\bibitem{higherVD1}
Pankaj~K. Agarwal, Mark de~Berg, Ji\v{r}\'{i} Matous\v{e}k, and Otfried
  Schwarzkopf.
\newblock Constructing levels in arrangements and higher order voronoi
  diagrams.
\newblock {\em SIAM Journal on Computing}, 27(3):654--667, 1998.

\bibitem{Agarwal1995}
Pankaj~K. Agarwal and Micha Sharir.
\newblock {\em Computer Science Today: Recent Trends and Developments}, chapter
  Algorithmic techniques for geometric optimization, pages 234--253.
\newblock Springer Berlin Heidelberg, Berlin, Heidelberg, 1995.

\bibitem{inverse_tree}
Behrooz Alizadeh and Rainer~E. Burkard.
\newblock The inverse 1-center location problem on a tree.
\newblock Technical Report 2009-03, Graz University of Technology, 2009.

\bibitem{weight_balancing}
Luis Barba, Otfried Cheong, Jean-Lou De~Carufel, Michael~Gene Dobbins, Rudolf
  Fleischer, Akitoshi Kawamura, Matias Korman, Yoshio Okamoto, J\'{a}nos Pach,
  Yuan Tang, Takeshi Tokuyama, Sander Verdonschot, and Tianhao Wang.
\newblock Weight balancing on boundaries and skeletons.
\newblock In {\em Proceedings of the Thirtieth Annual Symposium on
  Computational Geometry (SoCG 2014)}, pages 436--443, 2014.

\bibitem{reverse_NPhard}
Oded Berman, Divinagracia~I. Ingco, and Amedeo Odoni.
\newblock Improving the location of minimax facilities through network
  modification.
\newblock {\em Networks}, 24(1):31--41, 1994.

\bibitem{inverse_NPhard}
Mao-Cheong Cai, X.~G. Yang, and J.~Z. Zhang.
\newblock The complexity analysis of the inverse center location problem.
\newblock {\em Journal of Global Optimization}, 15(2):213--218, 1999.

\bibitem{parallel-sorting}
Richard Cole.
\newblock Parallel merge sort.
\newblock {\em SIAM Journal on Computing}, 17(4):770--785, 1988.

\bibitem{weighted_center}
Martin Dyer.
\newblock Linear time algorithms for two- and three-variable linear programs.
\newblock {\em SIAM Journal on Computing}, 13(1):31--45, 1984.

\bibitem{higherVDcomplexity}
Der{-}Tsai Lee.
\newblock On \emph{k}-nearest neighbor voronoi diagrams in the plane.
\newblock {\em {IEEE} Trans. Computers}, 31(6):478--487, 1982.

\bibitem{parametric}
Nimrod Megiddo.
\newblock Applying parallel computation algorithms in the design of serial
  algorithms.
\newblock {\em Journal of the ACM}, 30(4):852--865, 1983.

\bibitem{Megiddo-linear}
Nimrod Megiddo.
\newblock Linear programming in linear time when the dimension is fixed.
\newblock {\em Journal of the ACM}, 31(1):114--127, 1984.

\bibitem{Megiddo-balls}
Nimrod Megiddo.
\newblock On the ball spanned by balls.
\newblock {\em Discrete \& Computational Geometry}, 4(6):605--610, 1989.

\bibitem{higherVD2}
Edgar~A. Ramos.
\newblock On range reporting, ray shooting and \emph{k}-level construction.
\newblock In {\em Proceedings of the Fifteenth Annual Symposium on
  Computational Geometry (SoCG 1999)}, pages 390--399, 1999.

\bibitem{reverse_tree}
Jianzhong Zhang, Zhenhong Liu, and Zhongfan Ma.
\newblock Some reverse location problems.
\newblock {\em European Journal of Operational Research}, 124(1):77--88, 2000.

\end{thebibliography}

\end{document}